\documentclass[11pt]{article}
\usepackage{amsmath,amssymb}
\usepackage{authblk}
\usepackage{color}
\usepackage{mathrsfs}
\usepackage{mathptmx}
\usepackage{verbatim}
\usepackage{epsfig}
\usepackage{CJK}
\usepackage{bm}
\usepackage{amsfonts}
\usepackage{amsthm}
\input{epsf.sty}
\usepackage{epsfig}

\def\<{\langle}
\def\>{\rangle}

\def\be{\begin{equation}}
\def\ee{\end{equation}}
\def\ba{\begin{array}}
\def\ea{\end{array}}

\newtheorem{theorem}{Theorem}[section]

\newtheorem{proposition}{Proposition}[section]
\newtheorem{corollary}{Corollary}[section]
\theoremstyle{definition}
\newtheorem{remark}{Remark}[section]
\newtheorem{example}{Example}[section]
\newtheorem{definition}{Definition}[section]
\newtheorem*{definition32}{Definition 3.2$'$}

\newtheorem{problem}{Problem}
\numberwithin{equation}{section}

\usepackage{amsmath}
\usepackage{amsthm}
\usepackage{amsfonts}
\usepackage{pifont}
\usepackage{color}
\usepackage{bbm}
\usepackage{CJK}
\usepackage{mathrsfs}
\usepackage{fancyhdr}
\usepackage{graphicx}
\usepackage{pdfsync}
\usepackage{psfrag}
\usepackage{time}
\usepackage{txfonts}
\usepackage{tikz-cd}

\usepackage{stmaryrd}
\usepackage{mathrsfs}
\usepackage{txfonts}
\usepackage{amsfonts}

\usepackage{indentfirst}

\def\be{\begin{equation}}
\def\ee{\end{equation}}
\def\br{\begin{eqnarray}}
\def\er{\end{eqnarray}}

\title{On the Hamiltonian and Geometric structure of the Craik-Leibovich equation}

\author{$\mbox{Cheng Yang\thanks{
E-mail: c\_yang11@fudan.edu.cn} }$
\\
School of Mathematical Sciences,\ \  Fudan University\\
Shanghai  200433,\ \   P.R.China
\\
}

\begin{document}

\maketitle

\date{}

\begin{quote}
\small {\bf Abstract} In this paper we show that the Craik-Leibovich (CL) equation in hydrodynamics is the Euler equation on the dual of a certain central extension of the Lie algebra of divergence-free vector fields. From this geometric viewpoint, one can give a generalization of CL equation on any Riemannian manifold with boundary. We also prove a stability theorem for 2-dimensional steady flows of the Craik-Leibovich equation.\\
\small {\bf Keywords} Euler equation, Central extension, Hamiltonian structure, Craik-Leibovich equation, Langmuir circulation, Stability.
\end{quote}


\section{Introduction}
\par
In the present paper, we study the Hamiltonian and geometric structure of the Craik-Leibovich (CL) equation describes the averaged motion of incompressible fluid with fast oscillating boundary. We prove that the CL equation can be regarded as the Euler equation on the dual of an appropriate central extension of the Lie algebra of divergence-free vector fields. This geometric point of view allows one to generalize the CL equation to any Riemannian manifolds with boundary. Also, one can obtain the Hamiltonian formulation of this equation. Using the energy-Casimir method of Arnold, we prove a stability theorem for the steady planar flows of the CL equation.
\par
The CL equation was derived by Craik and Leibovich \cite{crle} in their study of the theoretical model for Langmuir circulation. In 1938, Langmuir \cite{lang} reported his observation of windrows of seaweeds in the Sargasso Sea. When a wind blows over a water surface\textemdash sea or lake\textemdash steadily, one can see that small material, like seaweeds or bubbles, floated on the water would align with the wind direction. This is related to Langmuir circulation. Since its discovery, this interesting phenomenon initiates a lot of researches, both experimental and theoretical.
\par
In Craik-Leibovich theory, these circulations are caused by the interaction between the fluid motion and its fast oscillating boundary.  The dynamics of the averaged fluid motion is described by the following CL equation in a 3-dimensional domain with boundary:
\begin{equation}\label{e01}
\left\{
\begin{array}{ll}
\frac{\partial v}{\partial t}+(v,\nabla)v+curl\;v\times V_s=-\nabla p,\\
(v+V_s)\cdot\mathbf{n}=0,
\end{array}
\right.
\end{equation}
where $V_s$ is a prescribed Stokes drift velocity related to the average of surface waves.
\par
Equation (\ref{e01}) was first derived by Craik and Leibovich using averaging method. In \cite{vla},  a new derivation of Craik-Leibovich equation was given by Vladimirov and his coauthors using the generalized Krylov-Bogolyubov averaging method combined with two-timing method. In \cite{yang}, the author develops a general perturbation theory on the reduced space of a principal $G-$bundle. Applying this theory to a principal $\text{SDiff}(M)-$bundle related to the free boundary problems of incompressible fluid, the author was able to obtain the CL equation.
\par
 The above perturbation theory studied by the author also illuminates us the geometric structure of the CL equation: it turns out to be the Euler equation on the dual of a certain central extension of the Lie algebra of divergence-free vector fields. This geometric point of view enables one to give a higher-dimensional generalization of the CL equation on any Riemannian manifold with boundary in any dimension. And a large class of invariant functionals follows from this geometric structure.
 \par
 From this geometric viewpoint, We also give the Hamiltonian structure of the generalized CL equation on any $n-$dimensional Riemannian manifold with boundary. The Hamiltonian formulation of the classical CL equation was studied in works of Holm \cite{holm} and Vladimirov \cite{vla}.
\par
Euler equations on the duals of central extensions of Lie algebras arise in many interesting settings in mathematical physics. We would like to mention the infinite conductivity equation and the $\beta-$plane equation.
\begin{example}
The infinite conductivity equation on a Riemannian manifold $M$ in $\mathbb{R}^3$ is
\begin{equation}\label{e02}
\frac{\partial v}{\partial t}=-(v,\nabla)v-v\times B-\nabla p,
\end{equation}
where $B$ is a constant divergence-free magnetic field. In \cite{khch} Khesin and Chekanov studied the Hamiltonian and geometric structure of this equation. It turns out that this equation can be seen as the Euler equation on the dual space of the central extension of the Lie algebra of the divergence-free vector fields $\text{SVect}(M)$. The corresponding 2-cocycle is a Lichnerowicz 2-cocycle (see section 3.2) related to the magnetic field $B$. In their paper, Khesin and Chekanov generalized the infinite conductivity equation to any Riemannian manifolds in any dimension and found a large class of invariant functionals.
\end{example}

\begin{example}
 Another interesting equation appears in the study of the fluid motion on a rotating surface. In \cite{zeit} Zeitlin considered the $\beta$-plane equation (or Rossby waves equation):
\begin{equation}\label{e03}
\dot{\omega}+\{\psi,\omega\}+\beta\psi_x=0,
\end{equation}
where $\beta$ is a constant related to the Coriolis force, $\omega$ and $\psi$ are the vorticity and stream functions of the velocity fields of the fluid motion. It is the Euler equation on the dual of a central extension of the Lie algebra of the symplectomorphism group.
\end{example}
\par
This paper is organized as follows. Section 2 are some preliminaries about the Arnold's framework for Euler equation and central extensions of Lie algebras. In section 3, first we give the generalized CL equation on any Riemannian manifold with boundary (Theorem \ref{t31}), then prove the central extension structure of CL equation (Theorem \ref{t33}) together with the Hamiltonian structure of it, we also obtain a large class of invariant functionals. Section 4 discusses the stability of 2-dimensional steady flows of the Craik-Leibovich equation.

{\section{Geometric Preliminaries}

\subsection{Arnold's framework for the Euler equation}
\par
In his seminal paper \cite{arn}, Arnold developed the general theory for the Euler equation describing the geodesic
flow on an arbitrary Lie group equipped with a one-sided invariant Riemannian metric.
\par
Consider a finite or infinite-dimensional Lie group $G$ with Lie algebra $\mathfrak{g}$, and the dual of the Lie algebra is $\mathfrak{g}^*$. There exists a natural Lie-Poisson structure on the dual $\mathfrak{g}^*$.
\begin{definition}\label{d11}
 The \textbf{natural Lie--Poisson structure} $\{\;,\;\}_{LP}:C^{\infty}(\mathfrak{g}^*)\times C^{\infty}(\mathfrak{g}^*)\rightarrow C^{\infty}(\mathfrak{g}^*)\;$ on the dual space $\mathfrak{g}^*$ is the Poisson bracket defined by
$$
\{f,g\}_{LP}(m):=\langle [df,dg], m\rangle,\;\text{for}\; m\in \mathfrak g^*\;\text{and}\;f,g\in C^{\infty}(\mathfrak{g}^*),
$$
where the differentials are taken at the point $m$, and
$\langle \cdot,\cdot\rangle$ is a natural pairing between Lie algebra and its dual.
\end{definition}
It is well-known (see e.g. \cite{khmi}) that the Hamiltonian equation corresponding to a function $H$ and the Lie--Poisson structure $\{\;,\;\}_{LP}$ on $\mathfrak{g}^*$ is given by
\begin{equation}\label{e11}
\frac{dm}{dt}=ad^*_{dH}m.
\end{equation}
For a quadratic Hamiltonian function on $\mathfrak{g}^*$, we have the Euler equation:
\begin{definition}\label{d12}
 The \textbf{Euler equation} on $\mathfrak{g}^*$ is an equation corresponding to the quadratic (energy) Hamiltonian
$H(m)=-\frac 12\langle \mathbb{I}^{-1}m,m\rangle$:
\begin{equation}\label{e12}
\frac{dm}{dt}=-ad^*_{\mathbb{I}^{-1}m}m,
\end{equation}
where $\mathbb{I}:\mathfrak{g}\rightarrow \mathfrak{g}^*$ is an inertia operator.
\end{definition}
\par
Now we define a right-invariant metric on the Lie group $G$. Suppose we have a fixed quadratic form $E=\frac 12 \langle v,\mathbb{I}v\rangle$ on the Lie algebra $\mathfrak{g}$, where $v\in\mathfrak g$ and $\mathbb I:\mathfrak g\rightarrow\mathfrak{g}^*$ is an inertia operator, One can obtain the right-invariant metric on the tangent space $TG$ of group $G$ by using right translation.
\par
Arnold in \cite{arn} proved that for such a group with a right-invariant metric, the corresponding geodesic flow can be described by the equation (\ref{e12}). This equation (\ref{e12}) coincides with the classical Euler equation of an ideal fluid for the group $G=\text{SDiff}(M)$ with the right-invariant $L^2$-metric. More detailed discussion can be found in \cite{arkh}.
\par
Next we recall the Arnold's framework for the classical Euler equation of an incompressible fluid. Let $M$ be an $n-$dimensional Riemannian manifold with a volume form $\mu$, the configuration space for the motion of an incompressible fluid fills the manifold $M$ is the volume-preserving diffeomorphism group $\text{SDiff}(M)$, and its Lie algebra $\text{SVect}(M)$ consists of all the divergence-free vector fields on $M$. The natural right-invariant Riemannian metric on $\text{SDiff}(M)$ is induced from the $L^2$ product of divergence-free vector fields on $M$.
\par
 The (regular) dual space $\mathfrak g^*=\Omega^1(M)/d\Omega^0(M)$ is the quotient space of all 1-forms on $M$ modulo all exact 1-forms on $M$. The Lie algebra coadjoint action on a coset in the dual space $\Omega^1(M)/d\Omega^0(M)$ is well-defined: it is the Lie derivative along a vector field $v\in\text{SVect}(M)$.
\par
 We have the following theorem:
\begin{theorem}\label{t11}{\cite{arkh}}
The Euler equation on the dual space $\mathfrak g^*=\Omega^1(M)/d\Omega^0(M)$ is
\begin{equation}\label{e13}
\partial_t\;[u]=-\mathcal L_v[u],
\end{equation}
where $[u]\in\Omega^1(M)/d\Omega^0(M)$ is a coset of 1-forms and the 1-form $u=v^b$ is induced from the vector field $v\in\text{SVect}(M)$ by lowering indices.
\end{theorem}

\subsection{Central extensions}
\begin{definition}\label{d13}
 A \textbf{central extension} of a Lie algebra $\mathfrak{g}$ by a vector space $V$ is a Lie algebra $\hat{\mathfrak{g}}=\mathfrak{g}\oplus V$ with the Lie bracket:
$$
[(X,u),(Y,v)]^{\wedge}=([X,Y],\;\widehat{\omega}(X,Y)),
$$
for a Lie algebra 2-cocycle $\widehat{\omega}:\mathfrak{g}\times\mathfrak{g}\rightarrow V$, which is a bilinear, antisymmetric form satisfies the cocycle identity:
$$
\widehat{\omega}([X,Y],Z)+\widehat{\omega}([Y,Z],X)+\widehat{\omega}([Z,X],Y)=0.
$$
\end{definition}
\begin{example}
Let $M$ be a compact manifold with a volume form $\mu$ and $\beta$ is a closed 2-form on $M$. The Lichnerowicz 2-cocycle $\widehat{\omega}_{\beta}$ on Lie algebra $\text{SVect}(M)$ of divergence-free vector fields on $M$ tangent to the boundary of $M$ is defined by
$$
\widehat{\omega}_{\beta}(X,Y)=\int_{M}\beta(X,Y)\;\mu.
$$
\par
Now consider the case of the infinite conductivity equation defined on any $n-$dimensional Riemannian manifold $M$:
\begin{equation}\label{e14}
\frac{\partial v}{\partial t}=-(v,\nabla)v-v\times B-\nabla p,
\end{equation}
 where $B$ is an $(n-2)-$vector field on $M$ which is a smooth section on $\wedge^{n-2}TM$ and the cross product of a vector field $X$ with $B$ is the vector field $X \times B = *(X\wedge B)$. We set the closed 2-form $\beta = (-1)^{n-2}i_B\mu$, the equation (\ref{e14}) can be viewed as the Euler equation on the dual space of the central extension of $\text{SVect}(M)$ with a Lichnerowicz 2-cocycle $\widehat{\omega}_{\beta}$ corresponding to 2-form $\beta$ \cite{khch} \cite{rog} \cite{viz}.
\end{example}

\section{Geometric and Hamiltonian structure of the Craik-Leibovich equation}
\subsection{A special central extension}
In this section, we are going to define a special 2-cocycle for a general Lie algebra $\mathfrak g$. The corresponding central extension relates to the geometric structure of Craik-Leibovich equation.
\begin{definition}\label{d14}
 We define the \textbf{shifted 2-cocycle} $\widehat{\omega}_{V_s}:\mathfrak{g}\times\mathfrak{g}\rightarrow \mathbb{R}$ on the Lie algebra $\mathfrak{g}$ for a fixed vector $V_s\in\mathfrak{g}$ by
\begin{equation}\label{e113}
 \widehat{\omega}_{V_s}(X,Y)=-\left\langle ad^*_{X}\;\mathbb{I}(V_s),Y\right\rangle,
\end{equation}
where $X,\;Y\in\mathfrak{g}$ and $\mathbb{I}$ is the inertia operator on $\mathfrak{g}$.
\end{definition}
\begin{remark}
 Because $-\left\langle ad^*_{X}\;\mathbb{I}(V_s),Y\right\rangle=-\left\langle \mathbb{I}(V_s),[X,Y]\right\rangle$, we know that $\widehat{\omega}_{V_s}$ is a trivial 2-cocycle, or 2-coboundary.
\end{remark}

\par
Let $\hat{\mathfrak{g}}_{V_s}$ be the central extension of the Lie algebra $\mathfrak{g}$ with the 2-cocycle $\widehat{\omega}_{V_s}$. First, we derive the Euler equation on $\hat{\mathfrak{g}}_{V_s}^*$,
\begin{proposition}\label{t11}
  The Euler equation on $\hat{\mathfrak{g}}_{V_s}^*$ corresponding to the quadratic (energy) Hamiltonian $H(m)=-\frac{1}{2}\langle\mathbb{I}^{-1}m,m\rangle$ is
\begin{equation}\label{e114}
 \frac{d}{dt}\;m=-ad^*_{\mathbb{I}^{-1}m}\;\{m-a\mathbb{I}(V_s)\}.
\end{equation}
\end{proposition}
\begin{proof}
 Since
\begin{align*}
 \left\langle ad^*_{(X,b)}(m,a),(Y,c)\right\rangle
=&\left\langle (m,a),([X,Y],\widehat{\omega}_{V_s}(X,Y))\right\rangle
=\left\langle m,[X,Y]\right\rangle+a\;\widehat{\omega}_{V_s}(X,Y)\\
=&\left\langle ad^*_X\;m,Y\right\rangle-\left\langle a \;ad^*_{X}\mathbb{I}(V_s),Y\right\rangle
=\left\langle ad^*_X\;m-a \;ad^*_{X}\mathbb{I}(V_s),Y\right\rangle,
\end{align*}
we get that the Euler equation on the dual space $\hat{\mathfrak{g}}_{V_s}^*$ of the central extension of $\mathfrak{g}$ is
$$
 \frac{d}{dt}\;m=-ad^*_{\mathbb{I}^{-1}m}\;\{m-a\mathbb{I}(V_s)\}.
$$
\end{proof}

\subsection{Central extension structure of the Craik-Leibovich equation}
Let $M$ be an $n-$dimensional Riemannnian manifold with boundary $\partial M$, recall that group $\text{SDiff}(M)$ is the group of all volume-preserving diffeomorphisms on $M$, the Lie algebra $\text{SVect}(M)$ is all the divergence-free vector field on $M$ tangent to the boundary $\partial M$, the regular dual space $\Omega^1(M)/d\Omega^0(M)$ of the Lie algebra is the space of cosets of 1-forms on $M$ modulo the exact 1-forms. First we give an $n-$dimensional generalization of the Craik-Leibovich equation.
\begin{theorem}\label{t31}
The \textbf{\textit{$n-$dimensional generalized Craik-Leibovich (CL) equation}} on the space $\Omega^1(M)/d\Omega^0(M)$ is
\begin{equation}\label{e34}
 \frac{d}{dt}\;[u]=-\mathcal{L}_{v+V_s}\;[u],
\end{equation}
where $v+V_s\in \text{SVect}(M)$, and $[u]=[v^b]\in\Omega^1(M)/d\Omega^0(M)$.
\end{theorem}
\begin{proof}
Let $u=v^b$, the equation (\ref{e34}) becomes
$$
\frac{d}{dt}\;u=-\mathcal{L}_{v+V_s}\;u+d\psi,
$$
Follow from the identities
$$
\mathcal{L}_v(v^b)=(\nabla_v v)^b+\frac 12 d\langle v,\;v\rangle
$$
and
$$
*(curl \;v\wedge V_s)=i_{V_s}i_{curl \;v}\mu=i_{V_s}dv^b=\mathcal{L}_{V_s}u,
$$
we obtain an equation which can be seen as the CL equation on an $n-$dimensional manifold $M$
\begin{equation}\label{e35}
 v_t+\nabla_v v+curl \;v\times V_s=-\nabla p,
\end{equation}
where $curl\;v$ is defined as an $(n-2)-$vector field, also condition $v+V_s\in \text{SVect}(M)$ gives us the boundary condition. When the dimension $n=3$, equation (\ref{e35}) is the classical CL equation.
\end{proof}
\begin{remark}
Note that velocity fields $v$ and $V_s$ do not have to be elements in the Lie algebra $\text{SVect}(M)$, but their sum $v+V_S$ is an element of $\text{SVect}(M)$ which gives the boundary condition.
\end{remark}
We show that the CL equation (\ref{e34}) is a Hamiltonian equation,
\begin{corollary}\label{t32}
Equation (\ref{e34}) is a Hamiltonian equation on coadjoint orbits in $\mathfrak{g}^*=\Omega^1(M)/d\Omega^0(M)$ with the Hamiltonian function $H=-\frac 12([u+V_s^b],\;\mathbb{I}^{-1}[u+V_s^b])$.
\end{corollary}
\begin{proof}
The Hamiltonian equation on coadjoint orbits in $\mathfrak{g}^*=\Omega^1(M)/d\Omega^0(M)$ is
$$
\frac{d}{ds}\;[u]=-\mathcal{L}_{\frac{\delta H}{\delta [u]}}\;[u].
$$
For Hamiltonian function $H=-\frac 12([u+V_s^b],\;\mathbb{I}^{-1}[u+V_s^b])$, the functional derivative $\frac{\delta H}{\delta [u]}=\mathbb{I}^{-1}[u+V_s^b])=v+V_s$, so we have equation (\ref{e34}).
\end{proof}

Now set $[u]'=[u+V_s^b]$, equation (\ref{e34}) becomes
\begin{equation}\label{e36}
\frac{d}{dt}\;[u]'=-\mathcal{L}_{\mathbb{I}^{-1}[u]'}\;\{[u]'-[V_s^b]\}.
\end{equation}
\begin{theorem}\label{t33}
The equation (\ref{e36}) is the Euler equation on the central extension of the Lie algebra $\mathfrak{g}=SVect(D)$ by means of the 2-cocycle
$$
\widehat{\omega}_{V_s}(X,Y)=-\left\langle\mathcal{L}_{X}\;V_s^b,Y\right\rangle
$$
associated to the vector field $V_s$.
\end{theorem}
\begin{proof}
We choose $a=1$, then this theorem follows from proposition \ref{t11}.
\end{proof}
\begin{remark}
 This is the Lichnerowicz 2-cocycle. Indeed, let the closed 2-form $\beta=-dV_s^b$, then the corresponding Lichnerowicz 2-cocycle is
 \begin{align*}
 \widehat{\omega}_{\beta}(X,Y)
=&\int_M-dV_s^b(X,Y)\;\mu
=-\int_M\langle i_X dV_s^b,Y\rangle\;\mu\\
=&-\int_M\langle\mathcal{L}_X V_s^b,Y\rangle\;\mu+\int_M\langle di_X V_s^b,Y\rangle\;\mu
=-\left\langle\mathcal{L}_{X}\;V_s^b,Y\right\rangle=\widehat{\omega}_{V_s}(X,Y),
\end{align*}
in the above equations $\int_M\langle di_X V_s^b,Y\rangle\;\mu=0$ since $Y\in\text{SVect}(M)$.
\end{remark}
The geometric structure of the CL equation gives us the first integrals similar to those invariants studied in \cite{khch} for the infinite conductivity equation.
\begin{corollary}\label{t34}
Equation (\ref{e34}) has
\par
(1) an integral $I(v)=\int_M u\wedge\;(du)^m$ for $u=v^b$ in the case of an odd $n=2m+1$,
\par
(2) infinitely many integrals
$$
I_f(v)=\int_M f\left(\frac{(du)^m}{vol_M}\right)\;vol_M
$$
in the case of an even $n=2m$, here $vol_M$ is the volume form on $M$.
\end{corollary}
\begin{proof}
The form of equation (\ref{e34}) shows that the moment $[u]$ moves along coadjoint orbits of the $\text{SDiff}(M)$-action corresponding to $v+V_s$. Because $\text{SDiff}(M)$-action coincides with the change of variables and preserves the volume form on $M$, the functionals (1) and (2) are invariant integrals for the flows of the CL equation.
\end{proof}
\begin{remark}
From the proof of this corollary, one can see that the moment $[u]=[v^b]$ is transferred by the flow corresponding to the velocity field $v+V_s$. For the CL equation we give the two equivalent definitions of the isovorticed fields corresponding to equation (\ref{e34}) and (\ref{e36}) respectively:
\begin{definition}\label{d31}
For equation (\ref{e34}), two vector fields $v$ and $v'$ are \textbf{isovorticed} if $curl\;v$ can be transferred to $curl\;v'$ by a volume-preserving diffeomorphism and satisfy the same boundary condition: $(v+V_s)\cdot\mathbf{n}=(v'+V_s)\cdot\mathbf{n}=0$.
\end{definition}
\begin{definition32}\label{d32}
For equation (\ref{e36}), two vector fields $v,\;v'\in\text{SVect}(M)$ are \textbf{isovorticed} if $curl\;(v-V_s)$ can be transferred to $curl\;(v'-V_s)$ by a volume-preserving diffeomorphism.
\end{definition32}
\end{remark}
\subsection{Steady flows of the Craik-Leibovich equation}
The vorticity equation of the incompressible CL flow corresponding to equation (\ref{e36}) is
\begin{equation}\label{e360}
\frac{\partial\omega}{\partial t}+\{v,\;\omega-curl\;V_s\}=0,
\end{equation}
where the vorticity field $\omega=curl\; v$. Note that the velocity field $v$ in the above equation is an element in the Lie algebra $\text{SVect}(M)$ which is differed from the solution of the CL equation (\ref{e01}) by a \mbox{shift $V_s$}, in other words, the velocity field $(v-V_s)$ is a solution of the CL equation (\ref{e01}). Therefore, the steady solution of equation (\ref{e360}) satisfies
\begin{equation}\label{e44}
\{v,curl\;(v-V_s)\}=0.
\end{equation}
Let us consider the following variatioanl problem:
\begin{problem}
Suppose that the central extension group $\widehat{\text{SDiff}(M)}$ corresponding to the central extension of Lie algebra $\text{SVect}(M)$ described in Theorem \ref{t33} exists. Let us fix a vector field $v_0\in\text{SVect}(M)$, find the critical points of the kinetic function $K(v)=\frac 12 \langle v,v\rangle$ on the set  $S=\{v\in\text{SVect}(M)\mid(v,1)\;\text{can be transferred to}\;(v_0,1)\;\text{by an adjoint action of group}\;\widehat{\text{SDiff}(M)}\}$.
\end{problem}
It turns out that the stationary solutions of the CL equation coincide with the critical points of this variational problem.
\begin{theorem}\label{t43}
The stationary solutions of the equation (\ref{e360}) which satisfy equation (\ref{e44}) coincide with the critical points of the varational problem 1.
\end{theorem}
\begin{proof}
Let $(u,b)$ be an arbitrary element in the vector space $\text{SVect}(M)\oplus \mathbb{R}$. The variation $\delta(v,a)$ of a field $(v,a)$ under the adjoint action of $(u,b)$ is given by
$$
\delta(v,a)=[(u,b),\;(v,a)]^{\wedge}=([u,\;v],-\langle V_s,[u,\;v]\rangle).
$$
Let $v\in\text{SVect}(M)$ be a critical point of Problem 1, then the first variation of $E$ taken at $v$ should be 0, hence we have
\begin{align*}
0=\delta E=&\left\langle (v,1),\delta (v,1)\right\rangle=\langle (v,1),(\{v,\;u\},-\langle V_s,[u,\;v]\rangle)\\
=&\left\langle (v,1),\{v,\;u\}\right\rangle-\langle V_s,curl(u\times v)\rangle
=\left\langle u,v\times curl\;u\right\rangle-\langle curl\;V_s,u\times v\rangle\\
=&\left\langle u,v\times curl\;v\right\rangle-\langle u,v\times curl\;V_s\rangle
=\left\langle u,v\times curl\;(v-V_s)\right\rangle,
\end{align*}
So we have $\{v,curl\;(v-V_s)\}=0$.
\end{proof}

\section{Stability theorem for 2-dimensional steady flows of the Craik-Leibovich equation}
\subsection{Stability theorem for equilibrium points on central extensions of Lie algebras}
Let $\hat{\mathfrak g}=\mathfrak g\oplus \mathbb{R}$ be a one-dimensional central extension of an arbitrary Lie algebra with 2-cocycle $\hat{\omega}$. We introduce the bilinear operation $B:\mathfrak g \times \mathfrak g\rightarrow \mathfrak g$ defined by
\begin{equation}\label{e401}
\langle[v_1,v_2],v_3\rangle=\langle B(v_3,v_1),v_2\rangle,
\end{equation}
where $v_i\in \mathfrak g,\;i=1,2,3$. By the isomorphism between Lie algebra $\mathfrak g$ and its dual $\mathfrak{g}^*$, we have the Euler equation on the Lie algebra $\mathfrak g$ \cite{arkh}:
\begin{equation}\label{e402}
\frac{dv}{dt}=B(v,v),
\end{equation}
where $v\in\mathfrak g$.
\par
We also define an operator $w:\mathfrak g\rightarrow \mathfrak g$ induced from the 2-cocycle $\hat{\omega}$ by $\hat{\omega}(u,v)=\langle w(u),v\rangle$ for any $u,\;v\in \mathfrak g$.
\begin{proposition}\label{t401}
The Euler equation on the central extension $\hat{\mathfrak g}$ is
\begin{equation}\label{e404}
\frac{dv}{dt}=B(v,v)+a\;w(v).
\end{equation}
So the equilibrium point $(v_e,a_e)\in\hat{\mathfrak g}$ satisfies
\begin{equation}\label{e403}
B(v_e,v_e)+a_e\;w(v_e)=0.
\end{equation}
\end{proposition}
\begin{proof}
Let us compute the bilinear operation $\hat{B}:\hat{\mathfrak g}\times\hat{\mathfrak g}\rightarrow\hat{\mathfrak g}$ of the central extension $\hat{\mathfrak g}$:
\begin{align*}
\left\langle \hat{B}((v_3,a_3),(v_1,a_1)),(v_2,a_2)\right\rangle
=&\left\langle [(v_1,a_1),(v_2,a_2)],(v_3,a_3)\right\rangle
=\left\langle ([v_1,v_2],\hat{\omega}(v_1,v_2)),(v_3,a_3)\right\rangle\\
=&\left\langle [v_1,v_2],v_3\right\rangle+a_3\hat{\omega}(v_1,v_2)
=\left\langle B(v_3,v_1)+a_3 w(v_1),v_2\right\rangle,
\end{align*}
where $(v_i,a_i)\in\hat{\mathfrak g},\;i=1,2,3$. So the Euler equation on $\hat{\mathfrak g}$
$$
\frac{d(v,a)}{dt}=\hat{B}((v,a),(v,a))
$$
becomes equation (\ref{e404}). (For conciseness, here we omit the second equation $\frac{da}{dt}=0$.)
\end{proof}
\par
The dual space $\hat{\mathfrak g}^*$ of Lie algebra $\hat{\mathfrak g}$ is foliated by the coadjoint orbits, the isomorphism between $\hat{\mathfrak g}$ and $\hat{\mathfrak g}^*$ gives a coadjoint foliation of the Lie algebra $\hat{\mathfrak g}$. Next we prove a stability theorem for the equilibrium points on $\hat{\mathfrak g}$.
\begin{theorem}\label{t41}
Assume that the equilibrium point $(v_e,a_e)\in\hat{\mathfrak g}$ is regular for the coadjoint foliation of the Lie algebra $\hat{\mathfrak g}$. Consider a test quadratic form $T\mid_{(v_e,a_e)}$:
\begin{equation}\label{e41}
T\mid_{(v_e,a_e)}(\xi)=\langle B(v_e,\zeta)+a_e\;w(\zeta),\;B(v_e,\zeta)+a_e\;w(\zeta)\rangle+\langle [\zeta,v_e],\;B(v_e,\zeta)+a_e\;w(\zeta)\rangle,
\end{equation}
where $\xi=B(v_e,\zeta)+a_e\;w(\zeta)\in\mathfrak g$. If for all $\xi\in\mathfrak g$ we have $T\mid_{(v_e,a_e)}(\xi)>0$ or $T\mid_{(v_e,a_e)}(\xi)<0$ , then the equilibrium solution $(v_e,a_e)\in\hat{\mathfrak g}$ of the Euler equation on the central extension $\hat{\mathfrak g}^*$ is Lyapunov stable.
\end{theorem}
\begin{proof}
As proved by Arnold (see e.g. \cite{arkh}), the second variation of the kinetic function $K(v)=\frac 12 \langle v,\;v\rangle$ on the leaf of this coadjoint foliation of $\hat{\mathfrak g}$ is:
\begin{equation}\label{e410}
2\delta^2K\mid_{(v_e,a_e)}(\xi)=\langle \hat{B}(v_e,\zeta),\;\hat{B}(v_e,\zeta)\rangle+\langle [\zeta,v_e],\;\hat{B}(v_e,\zeta)\rangle,
\end{equation}
where $\xi=\hat{B}(v_e,\zeta)\in\mathfrak g$. Note that the quadratic form $\delta^2 K$ does not change with respect to the different choices of $\zeta$, it depends only on $\xi=\hat{B}(v_e,\zeta)$.
\par
By the computation in Proposition \ref{t401}, we have $\hat{B}(v_e,\zeta)=B(v_e,\zeta)+a_e\;w(\zeta)$. Substitute this into (\ref{e410}) we obtain the test quadratic form (\ref{e41}). The Lyapunov stability of the equilibrium point $(v_e,a_e)$ follows from a revised Langrange's theorem in chapter \S II.3 of \cite{arkh}.
\end{proof}

\subsection{An a priori estimate for 2-dimensional steady flows of the CL equation}
Let $D$ be a 2-dimensional domain with boundary and $dA$ is an area form. First let us explain the notations. Velocity field $v_e$ satisfies equation (\ref{e44}), therefore it is a stationary solution of the equation (\ref{e360}), and function $\psi_e$ is the stream function of it. Then let function $\psi_e^*$ be the stream function of the shifted velocity field $v_e-V_s$.
\par
Now we can prove the following theorem which gives an a priori estimate for 2-dimensional steady flows of the CL equation:
\begin{theorem}\label{t42}
Assume that in a 2-dimensional domain $D$ with a area form $dA$, (i) there exists a function $F$ such that $\psi_e=F(\Delta\psi_e^*)$, (ii) also, there are two constants $c_1$ and $c_2$ for which
\begin{equation}\label{e46}
0<c_1\leq\frac{\nabla \psi_e}{\nabla \Delta \psi_e^*}\leq c_2<\infty.
\end{equation}
Let $\psi(x,y,t)=\psi_e+\widetilde{\psi}(x,y,t)$ be the stream function corresponding to a different solution of the CL equation such that $\oint_{\;\partial D}\nabla^{\perp}\psi\cdot dl=\oint_{\;\partial D}\nabla^{\perp}\psi_e\cdot dl$. Then we have the following inequality for the perturbation $\widetilde{\psi}=\widetilde{\psi}(x,y,t)$
\begin{equation}\label{e47}
\|\nabla \widetilde{\psi}\|^2_2+c_1\|\Delta \widetilde{\psi}\|^2_2\leq\|\nabla \widetilde{\psi}_0\|^2_2+c_2\|\Delta \widetilde{\psi}_0\|^2_2,
\end{equation}
where $\widetilde{\psi}_0=\widetilde{\psi}(x,y,0)$ and $\|\cdot\|^2_2$ stands for the square of the $L^2-$norm, which is $\|v\|^2_2=\iint_D(v,v)\;dA$ for a vector field $v$ and $\|f\|^2_2=\iint_D f^2\;dA$ for a function $f$.
\end{theorem}

\begin{proof}
Consider the leaf of the coadjoint foliation of $\widehat{\text{SVect}(D)}$ contains the equilibrium point $(v_e,1)$, this leaf consists of all the elements of the form $(v,1)\in\widehat{\text{SVect}(D)}$ where $v$ is divergence-free vector field isovorticed to the equilibrium field $v_e$ in the sense of Definition 3.1$'$.
\par
First, we claim that the second variation of the kinetic function $K(v)=\frac 12 \iint_D(v,v)\;dA$ on the leaf described above is
\begin{equation}\label{e45}
\delta^2K\mid_{v_e}(\xi)=\frac 12 \iint_D\left( (\xi,\;\xi)+\frac{\nabla \psi_e}{\nabla \Delta \psi_e^*}(curl\;\xi)^2\right)\;dA,
\end{equation}
where $\xi$ stands for a variation field at $v_e$.
\par
Next let us prove this claim. According to equation (\ref{e410}), we have
\begin{equation}\label{e450}
2\delta^2K\mid_{v_e}(\xi)=\iint_D ((\xi,\;\xi)+(\xi,[\zeta,v_e]))\;dA,
\end{equation}
where $\xi=B(v_e,\zeta)+w(\zeta)=B(v_e-V_s,\zeta)$. The last term of this equation is
\begin{equation}\label{e451}
\iint_D(\xi,[\zeta,v_e])\;dA=\iint_D(\xi,curl\;(\zeta\times v_e))\;dA=\iint_D(curl\;\xi,(\zeta\times v_e))\;dA.
\end{equation}
Since $v_e=\nabla^{\perp}\;\psi_e$ and $v_e-V_s=\nabla^{\perp}\;\psi_e^*$, we get
$$
curl\;\xi=\mathcal{L}_{\zeta}\Delta\psi_e^*=(\zeta,\nabla\Delta\psi_e^*),
$$
$$
\zeta\times v_e=\zeta\times(\nabla^{\perp}\;\psi_e)=(\zeta,\nabla\psi_e),
$$
hence, we have
\begin{equation}\label{e452}
\zeta\times v_e=\frac{\nabla \psi_e}{\nabla \Delta \psi_e^*}curl\;\xi.
\end{equation}
By equation (\ref{e450}), (\ref{e451}) and (\ref{e452}), we prove the claimed equation (\ref{e45}).
\par
The rest of the proof is similar to Arnold's original proof of the criteria of hydrodynamic nonlinear stability. One can refer to \cite{arkh}. For completeness of the paper, we give the detailed proof in the appendix.
\end{proof}

\section{Appendix: The detailed proof of Theorem \ref{t42}}
By the assumption of Theorem \ref{t42}, we have $\psi_e=F(\Delta\psi_e^*)$, let the function $P$ be its primitive, i.e. $P'=F$. Then, $P''(\Delta \psi_e^*)=\frac{\nabla \psi_e}{\nabla \Delta \psi_e^*}$, again by the assumption we have $c_1\leq P''(\omega)\leq c_2$ which gives
$$
c_1\frac{\tilde{\omega}^2}{2}\leq P(\omega+\tilde{\omega})-P(\omega)-P'(\omega)\tilde{\omega}\leq c_2\frac{\tilde{\omega}^2}{2}.
$$
This implies
\begin{equation}\label{e51}
\|\nabla \widetilde{\psi}\|^2_2+2\iint_D(P(\Delta\psi_e^*+\Delta\widetilde{\psi})-P(\Delta\psi_e^*)-P'(\Delta\psi_e^*)\Delta\widetilde{\psi})\;dA\geq\|\nabla \widetilde{\psi}\|^2_2+c_1\|\Delta \widetilde{\psi}\|^2_2,
\end{equation}
\begin{equation}\label{e52}
\|\nabla \widetilde{\psi}_0\|^2_2+2\iint_D(P(\Delta\psi_e^*+\Delta\widetilde{\psi}_0)-P(\Delta\psi_e^*)-P'(\Delta\psi_e^*)\Delta\widetilde{\psi}_0)\;dA\leq\|\nabla \widetilde{\psi}_0\|^2_2+c_2\|\Delta \widetilde{\psi}_0\|^2_2.
\end{equation}
Introducing a functional
$$
C(\widetilde{\psi})=\frac{\|\nabla \widetilde{\psi}\|^2_2}{2}+\iint_D\left(P(\Delta\psi_e^*+\Delta\widetilde{\psi})-P(\Delta\psi_e^*)-P'(\Delta\psi_e^*)\Delta\widetilde{\psi}\right)\;dA,
$$
then the LHS of (\ref{e51}) and (\ref{e52}) are $2C(\widetilde{\psi}(t))$ and $2C(\widetilde{\psi}(0))$ respectively. Therefore if we could prove
\begin{equation}\label{e53}
C(\widetilde{\psi}(t))=C(\widetilde{\psi}(0)),
\end{equation}
then the theorem follows immediately from (\ref{e51}), (\ref{e52}) and (\ref{e53}).
\begin{proof}[Proof of (\ref{e53})]
We construct the following invariant functional according to the conservation of the kinetic energy and vorticity:
$$
\Gamma(\psi)=\frac{\|\nabla \psi\|^2_2}{2}+\iint_D P(\Delta \psi^*)\;dA,
$$
where $\nabla^{\perp}\psi^*+V_s=\nabla^{\perp}\psi$. The first variation of $\Gamma$ at the equilibrium solution $\psi_e$ is
\begin{align*}
\delta\;\Gamma\mid_{\psi_e}(\widetilde\psi)=&\iint_D((\nabla\widetilde\psi,\;\nabla\psi_e)+P'(\Delta\psi_e^*)\Delta\widetilde\psi)\;dA\\
=&\iint_D(-\psi_e\Delta\widetilde\psi+P'(\Delta\psi_e^*)\Delta\widetilde\psi)\;dA+\oint_{\;\partial D}\psi_e\frac{\partial \widetilde\psi}{\partial n}\;dl,
\end{align*}
since $P'(\Delta\psi_e^*)=F(\Delta\psi_e^*)=\psi_e$ and $\oint_{\;\partial D}\psi_e\frac{\partial \widetilde\psi}{\partial n}\;dl=0$, we get $\delta\;\Gamma\mid_{\psi_e}(\widetilde\psi)=0$.
\par
Note that for another functional $\tilde{\Gamma}(\widetilde\psi):=\Gamma(\psi_e+\widetilde\psi)-\Gamma(\psi_e)$, we have
$$
\tilde{\Gamma}(\widetilde\psi(t))=\tilde{\Gamma}(\widetilde\psi(0)),
$$
and
$$
\tilde{\Gamma}(\widetilde\psi)=\delta\;\Gamma\mid_{\psi_e}(\widetilde\psi)+C(\widetilde\psi),
$$
these two equalities imply (\ref{e53}).
\end{proof}

\bigskip

{\bf  Acknowledgements.}
The author is grateful to Boris Khesin for many fruitful discussions, and would like to thank the Department of Mathematics at the University of Toronto for  hospitality during his visit. He also wants to thank Xiaoping Yuan for valuable suggestions and encouragement.

  \end{document}